\documentclass[11pt]{article}
\usepackage[margin=1in]{geometry}
\usepackage{amsmath}
\usepackage{amsthm}
\usepackage{amstext}
\usepackage{amssymb, amsfonts}
\usepackage[pagewise]{lineno}
\usepackage{algorithmic}
\usepackage{algorithm}
\usepackage{mathtools}
\usepackage[colorlinks]{hyperref}

\newif\ifpublic
\publicfalse 

\ifpublic
\usepackage[disable]{todonotes}
\else
\usepackage[colorinlistoftodos]{todonotes}
\presetkeys{todonotes}{size=\footnotesize,color=green!40}{}
\fi

\newtheorem{theorem}{Theorem}
\newtheorem{lemma}[theorem]{Lemma}
\newtheorem{cor}[theorem]{Corollary}
\newtheorem{claim}[theorem]{Claim}
\newtheorem{proposition}[theorem]{Proposition}

\theoremstyle{definition}
\newtheorem{remark}[theorem]{Remark}
\newtheorem{definition}[theorem]{Definition}

\numberwithin{equation}{section}
\numberwithin{figure}{section}
\numberwithin{theorem}{section}

\newcommand{\poly}{\mathrm{poly}}
\newcommand{\zo}{\{0,1\}}
\newcommand{\Z}{\mathbb{Z}} 

\newcommand{\V}{v} 
\newcommand{\n}{n} 
\newcommand{\q}{q} 
\newcommand{\m}{q^\star} 
\newcommand{\AppMaj}{\mathrm{ApprMaj}}
\newcommand{\CC}[1]{\mathrm{CC}_{#1}} 
\newcommand{\Hgame}{\mathcal{H}} 

\newcommand{\inbrace}[1]{\left\{ #1 \right\}}

\begin{document}

\title{Hypergraph Two-coloring in the Streaming Model\footnote{A part of this work was reported earlier in the conference paper, 
 {\em Streaming Algorithms for 2-Coloring Uniform Hypergraphs} by Jaikumar Radhakrishnan and Saswata Shannigrahi, WADS (2011): 667-678.}
}

\author{ Jaikumar Radhakrishnan\footnotemark[2]
	\and Saswata Shannigrahi \footnotemark[3]
	\and Rakesh Venkat \footnotemark[4]
	}
\date{}
\maketitle
\renewcommand{\thefootnote}{\fnsymbol{footnote}}

\footnotetext[2]{Tata Institute of Fundamental Research, Mumbai. E-mail: \texttt{jaikumar@tifr.res.in}}
\footnotetext[3]{Indian Institute of Technology, Ropar. E-mail: \texttt{saswata@iitrpr.ac.in}}
\footnotetext[4]{Hebrew University of Jerusalem, Israel. E-mail: \texttt{rakesh09@gmail.com}. Part of this work was done when the author was a graduate student at the Tata Institute of Fundamental Research, Mumbai.}

\renewcommand{\thefootnote}{\arabic{footnote}}

\begin{abstract}

We consider space-efficient algorithms for two-coloring $\n$-uniform
hypergraphs $H=(V,E)$ in the streaming model, when the hyperedges
arrive one at a time.  It is known that any such hypergraph with at
most $0.7 \sqrt{\frac{\n}{\ln \n}} 2^\n$ hyperedges has a
two-coloring~\cite{Radhakrishnan-Srinivasan}, which can be found
deterministically in polynomial time, if allowed random access to all the hyperedges throughout its execution.

\begin{itemize}
\item Let $s^D(\V, \q,\n)$ be the minimum space used by a
  deterministic one-pass streaming algorithm that on receiving an
  $\n$-uniform hypergraph $H$ on $\V$ vertices and $\q$ hyperedges produces a proper two-coloring of
  $H$. We show that $s^D(\n^2, \q, \n) = \Omega(\q/\n)$ for $\q \leq
  0.7 \sqrt{\frac{\n}{\ln \n}} 2^\n$, and $s^D(\n^2, \q, \n) = \Omega(\sqrt{\frac{1}{\n\ln \n}} 2^\n)$ otherwise. 

\item Let $s^R(\V, \q,\n)$ be the minimum space used by a randomized
  one-pass streaming algorithm that on receiving an $\n$-uniform
  hypergraph $H$ on $\V$ vertices and $\q$ hyperedges, either produces a proper two-coloring of $H$ with high probability (say, $0.99$), or declares
  failure otherwise.  We show that $s^R(\V, \frac{1}{10}\sqrt{\frac{\n}{\ln
      \n}} 2^\n, \n) = O(\V \log \V)$.

\end{itemize}

The above results are inspired by the study of the number $\m(\n)$, the
minimum possible number of hyperedges in a $\n$-uniform hypergraph that
is \emph{not} two-colorable. It is known that $\m(\n) =
\Omega(\sqrt{\frac{\n}{\ln \n}} 2^\n)$ and $\m(\n)=O(\n^22^\n)$.  The lower
bound (due to Radhakrishnan and
Srinivasan~\cite{Radhakrishnan-Srinivasan}) has a corresponding
algorithm to deterministically produce the two-coloring. Our first
result shows that no space-efficient deterministic streaming algorithm can
match the performance of the algorithm in
~\cite{Radhakrishnan-Srinivasan}; the second result shows that there
is, however, a space-efficient randomized streaming algorithm for the task. 

\end{abstract}


\clearpage

\section{Introduction}

A hypergraph $H=(V,E)$ is a set system $E$ defined on a finite
universe $V$ (called the vertex set). The sets in the set system are
called \emph{hyperedges}. We consider $\n$-uniform hypergraphs, that
is, hypergraphs each of whose hyperedges have $\n$ elements.  We say that a
hypergraph is two-colorable, or has Property B, if there is an
assignment of colors $\chi:V\rightarrow\{\mathrm{red},\mathrm{blue}\}$
to the vertex set $V$ such that every hyperedge $e \in E$ has a
vertex colored red and a vertex colored blue, that is, no hyperedge is
\emph{monochromatic}. Note that the hypergraph two-coloring problem can be viewed as a constraint satisfaction problem where the clauses are of a specific kind (the \textrm{Not-All-Equal predicate}). 

\smallskip

In the special case of graphs (which we may view as two-uniform
hypergraphs) two-colorability is easy to characterize and establish:
the graph is two-colorable if and only if it does not have an odd
cycle; one can find a two-coloring in linear time and $O(|V|)$ space
with random access to the input.

\medskip
In contrast, when $n>2$,  the situation is not as simple: determining if a $n$-uniform hypergraph is two-colorable is NP-hard even for $n=3$, and there has been a long line of work on characterizing hypergraphs that have Property B (see Section~\ref{sec:prior-work-prop-b} for a brief overview). The best known result in this context is due to Radhakrishnan and Srinivasan~\cite{Radhakrishnan-Srinivasan}, who showed that any $\n$-uniform hypergraph with at most $O(\sqrt{\frac{n}{\ln n}}) 2^n$ hyperedges is two-colorable. 

\medskip
While the two-coloring algorithm of Radhakrishnan and Srinivasan \cite{Radhakrishnan-Srinivasan}  works in polynomial time, it assumes that one has random access to the input by having it available completely in memory. In many practical situations, such a restriction on the algorithm is not desirable: memory is limited to a small fraction of the size of the input, which needs to be loaded and read in parts from a source or \emph{stream}. Such scenarios are commonly encountered in practice, for e.g. in routers and databases.  

\medskip
With such considerations in mind,  the seminal work of Alon, Matias and Szegedy~\cite{AMS} defined and initiated a study of the \emph{streaming model}, and this area  has seen a large body of work since. In particular, there has been a growing interest on exploring algorithms on graphs in the \emph{semi-streaming} model. Proposed by Muthukrishnan~\cite{Muthukrishnan05}, this model looks at solving fundamental problems on graphs such as Max-Cut, $s-t$-connectivity, shortest paths etc. (see, for e.g. the works \cite{FeigenbaumKMSZ04, FeigenbaumKZ04, FeigenbaumKMSZ05, McGregor14} and references therein) when the edges are revealed sequentially in a stream one at a time. The algorithm only has $O(|V| \poly\log(|V|)$ bits of workspace, and limited number of passes over the edge-stream, which means that it cannot store the entire graph in memory.  For the specific problem of checking if the graph is bipartite (two-colorable), if  the edges are streamed one at a time, a two-coloring can be constructed with $O(|V| \log |V|)$ space and $O((|E| + |V|)\log^2 |V|)$ bit operations \cite{FeigenbaumKMSZ04}.

\medskip
In the case of $n$-uniform hypergraphs, the problem of having a limited workspace is further exacerbated: the number of hyperedges (and hence the input size) can be
\emph{exponentially} larger than the number of vertices. For example, the number of hyperedges can be as high as $\Omega(2^\n)$ even when there
are just $O(\n)$ vertices, and it would be unrealistic to expect the entire hypergraph to be available in the main memory. In line with the semi-streaming model for graphs, the corresponding model developed for hypergraphs is often referred to as \emph{set-streaming} in literature. Here, the hyperedges are viewed as coming in one at a time, represented as sets of vertices.  This model has has seen various applications in databases and social media analysis. Problems such as Max-Coverage, Hitting-Sets, and Independent Set have been studied \cite{SahaGetoor09,Halldorsson10,AlistarhIV15,McgregorVU16,EmekRosen16, ChitnisEtal16}, with the goal of these algorithms being to find good approximations to the optimal solution while using space that is a poly-logarithmic multiplicative factor of the vertex-set size:  $O(|V| \poly \log (|V|)$ bits. Perhaps surprisingly, despite being a fundamental combinatorial question having a rich background of investigation in the offline setting, the problem of \emph{two-coloring} has not been studied before in this model.   In this paper, we investigate the question of whether there are space-efficient streaming algorithms for two-coloring hypergraphs.

\medskip
Let us describe the set-streaming model in more detail. We can think of the hypergraph as being stored in external memory, as a sequence of hyperedges.  An algorithm with limited local work space analyzes it by making a small number of sequential passes over the external memory; this can be equivalently seen as a stream of hyperedges being presented to the algorithm in some (adversarial) sequential order.  After making a pass over the complete data set, the algorithm must decide as fast as possible either to start another pass or stop with an output. The resources whose use one tries to minimize in this model are: the number of passes over the data, the amount of workspace memory required (in bits)\footnote{Henceforth, an algorithm requiring $s$ space will refer to it requiring $s$ bits of workspace.}, and the maximum processing time for any data item.  The algorithm could be either deterministic or randomized, and in the latter case, its output needs to be accurate with high probability over the algorithm's internal coin tosses. 

\medskip

The parameter of interest for us is the \emph{uniformity} $\n$ of the hypergraph, and we measure the workspace requirements of our streaming algorithms  in terms of $\n$.  We think of the number of vertices as being  $\poly(\n)$; this is the most interesting regime for Property B. Note that any two-coloring requires $|V|$ bits to be described completely, and hence the space requirement  of any streaming algorithm for two-coloring a hypergraph is trivially $\Omega(|V|)$ bits. If the number of vertices is  larger, our bounds still hold, but the contrast between the trivial workspace-requirements and the input size is smaller and hence less interesting.  The hyperedges are made available one at a time as a list of $\n$ vertices each, where every vertex is represented using $B$ bits. We regard a streaming algorithm as being space-efficient if it utilizes a workspace that is $\poly(n)$, or equivalently, $\poly(|V|)$ bits. While this may seem more relaxed than just the $|V| \poly\log(|V|)$ space which is often used in the set-streaming model, our randomized algorithms actually work with $O(|V| \log |V|)$ space, and our lower bound on the space requirements of deterministic algorithms is almost linear in the number of hyperedges, which is close to optimal and could be exponential in $n$. 

\subsection{Prior Work on Property B}\label{sec:prior-work-prop-b}

Property B  for $n$-uniform hypergraphs has been studied in the offline setting extensively; we review some of the relevant literature here. Most of these works deal with combinatorially characterizing  when a $n$-uniform hypergraph is two-colorable. Using the probabilistic method, Erd\H{o}s~\cite{Erdos} showed that any $\n$-uniform hypergraph with fewer than $2^{\n-1}$ hyperedges is two-colorable: a random two-coloring is valid with positive probability; furthermore, this randomized method can be derandomized using the method of conditional probabilities.  Erd\H{o}s~\cite{Erdos64}
later showed that there are $\n$-uniform hypergraphs with $\Theta(\n^2)$
vertices and $\Theta(\n^{2}2^{\n})$ hyperedges that are \emph{not}
two-colorable.  Both these bounds remained unchanged for some time,
until it was first improved by Beck~\cite{Beck}, and further improved
by Radhakrishnan and Srinivasan~\cite{Radhakrishnan-Srinivasan}, who
showed that any hypergraph with fewer than $0.7\sqrt{\frac{\n}{\ln
    \n}}2^{\n}$ hyperedges is two-colorable (provided $\n$ is larger than some constant; for all $\n \geq 2$, their proof yields a bound of $\frac{1}{10}\sqrt{\frac{\n}{\ln\n}}2^{\n}$).  They also provided a
polynomial-time randomized algorithm (and its derandomization) for
coloring such hypergraphs.  The algorithm has \emph{one-sided} error,
i.e., with some small probability $\delta$ that can be made
arbitrarily small, it declares failure, but otherwise it always
outputs a valid two-coloring.  Recently, Cherkashin and
Kozik~\cite{CherkashinK15} showed that a remarkably simple
randomized algorithm achieves this bound. Erd\H{o}s and Lov{\'a}sz
\cite{Erdos-Lovasz} have conjectured that any hypergraph with fewer than
$\n2^{\n}$ hyperedges is two-colorable.  In related work, Achlioptas et al.~
\cite{Achlioptas} studied when, in terms of the number of hyperedges,
does a \emph{randomly} chosen $\n$-uniform hypergraph stop being
two-colorable, and relate this to similar questions for random $\n$-SAT. The above results are formally stated in terms of bounds on the number $\m(\n)$ \footnote{We use the slightly non-standard designation $\m(\n)$, as opposed to the more common $m(\n)$, to avoid conflict with some natural notations in our proofs.}, the minimum number of hyperedges in a \emph{non-two-colorable} $n$-uniform hypergraph.

\section{Our results}
\subsection{Deterministic streaming algorithms:} 

In the first part of the paper, we investigate the space requirements
of \emph{deterministic} streaming algorithms for the hypergraph
two-coloring problem. It seems reasonable to conjecture that any
deterministic one-pass algorithm must essentially store
all the hyperedges before it can arrive at a valid coloring.  Let
$s^D(\V, \q,\n)$ be the minimum space used by a deterministic
(one-pass) streaming algorithm that on receiving an $n$-uniform
hypergraph $H$ on $\V$ vertices and $\q$ hyperedges produces a proper
two-coloring of $H$ (if it exists). Recall that any hypergraph with fewer than $\m(\n)$ hyperedges is two-colorable.
\bigskip

\begin{theorem} \label{thm:result1}
$s^D(\n^2, \q, \n) = \Omega \left( \min \{\q, \m(n)\} / \n \right) $.
\end{theorem}
In particular, when $\q \leq 0.7 \sqrt{\frac{\n}{\ln \n}} 2^\n$, the hypergraph is
guaranteed to have a two-coloring. However, the above theorem shows if
the number of hyperedges is large, a two-coloring (though it is guaranteed
to exist) cannot be found by a space-efficient deterministic streaming
algorithm. For hypergraphs where the number of hyperedges exceeds $\m(\n)$, a deterministic streaming algorithm will take at least $\m(\n)$ space to color the hypergraph properly, and this is known to be at least exponential in $n$. 

Lower bounds for space bounded computations often follow from lower
bounds for associated communication complexity problems. The above
result is also obtained using this strategy. However, the
communication complexity problem turns out to be somewhat subtle; in
particular, we are not able to directly reduce it to a well-known
problem and refer to an existing lower bound. We conjecture that no
deterministic algorithm can do substantially better if it is allowed
only a constant number of passes over the input. Note, however, that
proving this might be non-trivial, because the corresponding two-round
deterministic communication problem has an efficient protocol; see
Section~\ref{sec:detub} for details.

\subsection{Randomized streaming algorithms:}
We show that a version of the delayed recoloring algorithm can be
implemented in the streaming model, and provides essentially the same
guarantees as the original algorithm.  Let $s^R(\V, \q,\n)$ be the
minimum space used by a randomized (one-pass) streaming algorithm that
on receiving an $\n$-uniform hypergraph $H$ on $\V$ vertices and $\q$
hyperdges with probability at least $\frac{3}{4}$ (say) produces a proper
two-coloring of $H$ (or declares failure). We suppose that each vertex
is represented using $B$ bits, and each hyperedge is represented as an
$\n$-tuple of vertices.

\begin{theorem}
\label{thm:RandomizedOnePass}
$s^R(\V, \frac{1}{10}\sqrt{\frac{\n}{\ln \n}} 2^\n, \n) =
O(\V B)$. Furthermore, the corresponding randomized algorithm maintains
a coloring of the vertices encountered, and updates this coloring in
time $O(\n \V B)$ per hyperedge.  If the hypergraph has at most
$\frac{1}{10}\sqrt{\frac{\n}{\ln \n}} 2^\n$ hyperedges, then with probability at least $3/4$, the two-coloring is valid \footnote{This can be amplified to any constant, with a corresponding constant factor increase in the space requirement.}.  If the two-coloring being
maintained is invalid, the algorithm declares failure (the algorithm never
outputs an invalid coloring).
\end{theorem}

In a few special cases, we can two-color hypergraphs with more hyperedges using space-efficient streaming algorithms. 

\begin{theorem} \label{thm:vertices-bounded}
 Let $t$ be such that $4 \leq t \leq \frac{\n^2}{2\n-1}$. An $\n$-uniform hypergraph $H$ with at most $\n^2/t$ vertices and
  at most $2^{\n-1} \exp(t/8)$ hyperedges is two-colorable. If the number of vertices in $H$ is known in advance, then the
  corresponding two-coloring can be found using an efficient one-pass
  randomized streaming algorithm: $s^R \left(\frac{\n^2}{t},2^{\n-1}\exp(t/8), n\right) = O(\n^2/t)$.
\end{theorem}

The above theorem is useful when the number of vertices is bounded. For instance, setting $ t = 8\ln (2\n) $ in Theorem~\ref{thm:vertices-bounded} implies that for any $\n$-uniform hypergraph with $|V|
\leq \frac{\n^2}{8 \ln (2\n)}$ vertices and $\q < \n^2 2^\n$ hyperedges, there is a randomized
one-pass streaming algorithm that outputs a two-coloring with high
probability.  This algorithm requires $O(\frac{\n^2}{\ln \n})$ space at
any instant and $O(\n)$ processing time after reading each hyperedge.
\medskip

\begin{theorem} \label{thm:LLL-streaming}
If each hyperedge of a hypergraph $H$ intersects at most
$\frac{(1-\epsilon) 2^{\n-1}}{e} -1$ other hyperedges, then a
two-coloring of the hypergraph can be found by a randomized streaming algorithm
that makes at most $O(\log |V|)$ passes over the input stream and uses
space $O(|V|B)$. 
\end{theorem}

The streaming algorithm above takes $O(\log |V|)$ passes over the input; but there is no restriction on the number of hyperedges, as long as they have bounded intersections. We use the algorithmic version of the Lov{\'a}sz Local Lemma given by Moser and Tardos~\cite{Moser-Tardos} to prove this.

%
%
%
%


\paragraph{Organisation of the rest of the paper:}

In Section~\ref{sec:Notation}, we introduce the notation.  In
Section~\ref{sec:deterministiclower}, we establish
Theorem~\ref{thm:result1}, showing the limitation of deterministic
streaming algorithms.  In Section~\ref{sec:streaming-alg}, we prove Theorem~\ref{thm:RandomizedOnePass} by 
showing how a version of the off-line delayed recoloring algorithm of
\cite{Radhakrishnan-Srinivasan} (see Section~\ref{sec:off-line-delayed-recoloring}) can, in fact, be
implemented efficiently in the streaming model.  We also give the the algorithms for Theorem~\ref{thm:vertices-bounded} and Theorem~\ref{thm:LLL-streaming} in this section. We conclude with some remarks and open problems.

\section{Notation}

\label{sec:Notation}

We denote $\n$-uniform hypergraphs by $H=(V,E)$, where $V$ is
the set of vertices, and $E\subseteq{V \choose \n}$ is the set of
hyperedges of $H$. For a hypergraph $H=(V,E)$, we use $\V$ for
$|V|$, and $\q$ for $|E|$ when these terms appear in algebraic expressions.  In our setting, typically, $\V$ will be a
small polynomial in $\n$ and $\q$ will be exponential in $\n$. For any $k\in
\mathbb{N}$, we use the notation $[k] \coloneqq \{1, \ldots, k\}$.

A valid \emph{two-coloring} of the hypergraph $H$ is an assignment
$\chi:V\rightarrow\{\mathrm{Red}, \mathrm{Blue}\}$ that leaves no hyperedge
monochromatic, i.e., $\forall e\in E$, $\exists i,j\in e$ such that
$\chi(i)\neq\chi(j)$.  A hypergraph that admits a valid two-coloring
is said to be \emph{two-colorable,} or equivalently, to have
\emph{Property B}. 

We are interested in space-efficient streaming algorithms for
finding two-colorings of hypergraphs.  Consider a hypergraph with hyperedge
set $E=\{e_{1},\ldots,e_{\q}\}$. The hyperedges are made available to the
algorithm one at a time in some order $e_{i_1},\ldots, e_{i_\q}$.  To
keep the problem general, we do not assume that the vertex set is
fixed in advance unless otherwise stated; the algorithm will become aware of the vertices as
they arrive as part of the stream of hyperedges. We assume that
each vertex is encoded using $B$ bits.  The goal is to design a
space-efficient algorithm (deterministic or randomized) that can
output a valid coloring for the entire hypergraph once all the
hyperedges have \emph{passed}. We may allow multiple passes over the
input data. We call the algorithm a $r$-pass streaming algorithm, if
it outputs a valid coloring after making $r$ passes over the input
stream. By space-efficient, we mean that the algorithm uses
$\poly(\n)$ bits of internal workspace. We assume that $\n$ is
large, say, at least 100.

\smallskip
Denote by $ s^D_r (\V,\q,\n)$ the minimum space used by a deterministic $r$-pass
streaming algorithm that on receiving an $\n$-uniform hypergraph $H$ on $\V$ vertices and $\q$ hyperedges
produces a proper two-coloring of $H$. Similarly, $ s^R_r (\V,\q,\n)$ is the minimum space used by a randomized $r$-pass streaming algorithm that on receiving an
two-colorable $\n$-uniform hypergraph $H$ on $\V$ vertices and $\q$ hyperedges, with probability at least $\frac{3}{4}$ produces a
proper two-coloring of $H$ or declares failure otherwise. Note that such an algorithm can be made to output a valid coloring with probability $1-\delta$, for any $\delta \in (0,1)$ by running it in parallel $O(\log (1/\delta))$ times and choosing any one of the valid colorings produced. This increases the space requirement by a multiplicative factor of $O(\log (1/\delta))$.  When $r$ is omitted from the subscript, it is assumed that $r=1$.

\section{Deterministic streaming algorithms}
\label{sec:deterministiclower}
We first show lower bounds in the deterministic setting. We recall the
routine translation of an efficient streaming algorithm to a
communication complexity protocol~\cite{AMS}, with a view to proving
lower bounds. There are two computationally unbounded players Alice
and Bob, who both know of a relation $R \subseteq \mathcal{X} \times
\mathcal{Y} \times \mathcal{Z}$.  Alice receives an input $x \in
\mathcal{X}$ and Bob gets $y \in \mathcal{Y}$; in the beginning,
neither player is aware of the other's input. Their goal is to exchange
bits according to a fixed \emph{protocol} and find a $z \in
\mathcal{Z}$ so that $(x,y,z)\in R$. The communication complexity of
$R$ is the minimum number of bits that Alice and Bob exchange in any
valid protocol for the worst case input pair $(x,y)$. Several generalizations of this model can be defined with $k \geq 3$ players, we will define our specific model below. For more
details on communication complexity in general, please consult the
book by Kushilevitz and Nisan~\cite{Kushilevitz-Nisan}.

\smallskip
To show lower bounds for hypergraph two-coloring, we define the class of communication problems $\Hgame(\V, \q , k)$ \footnote{Since all hypergraphs we consider are $\n$-uniform, we will not explicitly state $\n$ as a parameter}.

\begin{definition}[Problem class $\Hgame(\V, \q, k)$]  \label{defn:hyper-comm-game}

For $k\geq 2$, an instance $I \in \Hgame(\V, \q, k)$ has $k$ players
$P_{1},\ldots,P_{k}$. Each player $P_{i}$ has a subset
$E_{i}\subseteq E$ of some hypergraph $H=(V,E = E_1\cup E_2\cup \cdots \cup E_k)$, with $|V| = \V$ and $|E_i|\leq q$. The communication is done in sequential order:
starting with $P_{1}$, $P_{i}$ sends a message to $P_{i+1}$  for $i\in \{1,\ldots, k-1\}$. This sequence of communication constitutes a round. In a multiple-round protocol, $P_{k}$ may start a new round by sending back a message to $P_1$, who would continue communication in the above order. In a valid $r$-round protocol $\Pi$, some $P_i$ in the course of the $r$-th round will output a coloring  $\chi$ that is valid for $H$. Denote by $\Pi(I, P_i, l)$ the communication sent by $P_i$ in round $l$ on instance $I$. We define the $r$-round communication complexity of  $\Hgame(\V, \q, k)$ as follows:

$$\CC{r}(\Hgame(\V, \q, k))  = \min_{\substack{ \Pi:~r-\text{round} \\ \text{ valid protocol}}} ~~ \max_{\substack{I\in \Hgame(\V,\q ,k),\\ i
\in [k], l\in[r]}} 
\quad \left\vert \Pi(I, P_i, l)  \right\vert $$  

We will always assume that $k\cdot\q \leq \m(\n)$ for this problem class, for technical reasons. This ensures that the hypergraph $H$ is always two-colorable.
\end{definition}

\begin{remark} \rm In this definition, we consider the maximum
communication by any single player instead of total communication because this quantity is
related more closely to the space requirement of streaming protcols
(see Proposition~\ref{prop:streamingtocommunication} below). If the maximum communication by any player is $s$ bits, then the total communication is bounded above by $kr\cdot s$ bits for an $r$ round protocol. Besides, for the values of $k$ and $r$ that we consider, this $O(kr)$ blowup is immaterial.
\end{remark}

\begin{remark}
Our lower-bound proofs requires that the total number of hyperedges is at most $\m(\n)$, as our arguments assume that the hypergraphs are always two-colorable. Note that this is known to be exponential in $n$, so our lower bound is always non-trivial.
\end{remark}

Note that when $k=2$ and  $r=1$,  we
get the two-player one-round model (where we traditionally call $P_1$  as Alice and $P_2$ as Bob): Alice sends a message $m$ to Bob
depending on her input and Bob outputs a coloring looking at $m$ and
his input. Our lower bound for hypergraph
coloring in the streaming model will rely on the following connection
between streaming and communication complexity, which follows from the well-known observation introduced in \cite{AMS}.

\begin{proposition}[\cite{AMS}] \label{prop:streamingtocommunication}
For any $k \in \Z, ~k \geq 2$, we have:
$$s^D_{1}(\V, \q, n) \geq  \CC{1}(\Hgame(\V, \min \inbrace{\q, \m(\n)}, k))$$
\end{proposition}
\begin{proof}
When $\q \leq \m(\n)$, the lower bound follows directly from the corresponding observation in \cite{AMS}. When $\q \geq \m(\n)$, the lower bound when $\q = \m(\n)$ still applies, as the algorithm does not know the number of hyperedges in advance, and hence has to two-color the sub-hypergraph on the first $\m(\n)$ hyperedges correctly.
\end{proof}

Thus, to establish Theorem~\ref{thm:result1}, it is enough to show an
appropriate lower bound on $\CC{1}(\Hgame(\V,\q,k))$.
We start with the two-player case, which already introduces most
of the ideas.

\smallskip

\begin{theorem}
\label{thm:OneRoundLowerBound} 
Let $\q \leq \m(\n)/2$. Then, $\CC{1}(\Hgame(\n^2, \q, 2))= \Omega(\frac{\q^{2}}{2^{\n}\n^{2}})$.
\end{theorem}

\begin{remark}\rm Note that the above theorem gives a non-trivial lower
  bound only when $\q \gg \n 2^{\n/2}$. Ideally, we would expect a lower bound that is linear in $\q$, for all values of $\q$.  
\end{remark}

\begin{proof} Consider a valid one-round protocol for the two-player
problem where Alice sends a message $m$ from a set $\mathcal{M}$ of
possible messages.  Let the input to the protocol be $(H_A,H_B)$,
where $H_A$ and $H_B$ are hypergraphs, each with $\q$ hyperedges on a
common vertex set $[\n^2]$.  Note that every possible hypergraph $H = (H_A \cup H_B)$ has less than $\m(\n)$ hyperedges, and hence has a valid two-coloring. For every hypergraph $H_{B}$ that Bob
receives, he must output a coloring $\chi=f(m,H_{B})$ based on some
deterministic function $f$. For $m\in \mathcal{M}$, define
\[
 L(m)=\{f(m, H_B): \text{$H_B$ is an input for Bob}
\} \, ,
\]
which is a list of two-colorings of the vertex set $V$. It is easy to see that Bob may identify the message $m$ with the list
$L(m)$; on receiving $m$, he must find a proper coloring for $H_B$
from $L(m)$.   Thus, for every $m\in\mathcal{M}$, we have the
following.
\begin{description}
\item[Completeness for Bob:] Every possible input hypergraph to Bob (i.e. all
hypergraphs on $\q$ hyperedges) should have a valid coloring
in $L(m)$. 

\item[Soundness for Alice:] Let $P(m)=\{H_A: \text{Alice sends the message }m \text{ for 
input } H_A \}$.
Then every $\chi \in L(m)$ should be valid for every hypergraph $H_{A} \in 
P(m)$.
\end{description}

We will show that these two conditions imply the claimed lower bound
on $|{\cal M}|$.

\begin{definition}[Shadows] Given a coloring $\chi$, we define its
  \emph{shadow} $\Delta(\chi)$ to 
be the set of all possible hyperedges that are monochromatic under $\chi$. The
shadow of a list $L$ of colorings is
\[            \Delta(L)=\bigcup_{\chi\in L} \Delta(\chi)
\]
\end{definition}
Note that in the above definition, the shadow $\Delta(\chi)$ collects
\emph{all possible} monochromatic hyperedges under $\chi$, so it
depends only on the coloring $\chi$, and not on any
hypergraph. Similarly, $\Delta(L)$ is also a collection of hyperedges
and does not depend on any hypergraph; in particular, if a hypergraph
$H_A$ is monochromatic under every coloring in $L$, then none of $H_A$'s
hyperedges can appear in $\Delta(L)$. 

In the following, assume that $\n$ is large.
\begin{claim}
For every coloring $\chi$, we have 
\begin{equation}
\Delta(\chi) 
\geq \frac{1}{10 \cdot 2^{\n}} {\n^2 \choose \n}. \label{eq:oneshadow}
\end{equation}
\end{claim} 
\begin{proof}[Proof of Claim]: One of the two color classes $\chi$ has at least
$\lceil{\n^2/2}\rceil$ vertices. It follows that 
\[ \Delta(\chi) \geq {\lceil{\n^2/2}\rceil \choose \n} \geq \frac{1}{10
  \cdot 2^{\n}} {\n^2 \choose \n}. \quad \quad
\]
\qquad 	\end{proof}

We next observe that the {\em completeness condition} for Bob imposes a
lower bound on $\Delta(L(m))$. 
\begin{claim} For every $m\in \mathcal{M}$,
\label{cl:ShadowUnion} 
\begin{equation}
|\Delta(L(m))| \geq \frac{\q}{10\n^2 2^{\n}} {\n^2 \choose \n}. \label{eq:ShadoUnion}
\end{equation}
\end{claim}
\begin{proof}[Proof of Claim]:
Suppose the claim does not hold, that is, 
\begin{equation}
|\Delta(L(m))| < \frac{\q}{10\n^22^{\n}}{\n^2 \choose \n}. \label{eq:shadowassumption}
\end{equation}
Choose a random hypergraph $H_{B}$ by choosing $\q$ hyperedges randomly
from $\Delta(L(m))$. We say that the hypergraph \emph{hits}
$\chi\in L(m)$, if at least one of its hyperedges falls in $\Delta(\chi)$,
otherwise we say it \emph{misses} $\chi$. For each $\chi \in L(m)$, we
have, using the bounds (\ref{eq:oneshadow}) and
(\ref{eq:shadowassumption}), that
\begin{align*}
  \Pr_H[H \text{ misses } \chi] & \leq 
\left(1-\frac{|\Delta(\chi)|}{|\Delta(L(m))|}\right)^\q \\
				& \leq \left(1-\frac{\n^2}{\q}\right)^\q \\
				& \leq \exp(-\n^2)
\end{align*}
Since there are at most $2^{\n^2}$ colorings, the union bound yields: 
\[
  \Pr_H[\exists \chi\in L(m): H \text{ misses } \chi] \leq 2^{\n^2}
  \exp(-10 \n^2) \ll 
1 
\]
Thus there exists a hypergraph $H$ with $\q$ hyperedges that hits every coloring
$\chi \in L(m)$, that is every coloring in $L(m)$ is invalid for $H$.
This, however, violates Completeness for Bob, proving the Claim. 
\end{proof}

We can now complete the proof of the theorem.  Consider a random
hypergraph $H$ for Alice, obtained by choosing each of its hyperedges
uniformly at random from the set of all hyperedges. Since Alice sends
some $m\in\mathcal{M}$ for \emph{every} hypergraph, the {\em soundness
  condition} for Alice implies that
\[
    1 = \Pr_H[\exists m: H \text{ misses all } \chi \in L(m)] \leq 
|\mathcal{M}| (1-\frac{\q}{10\n^22^\n})^\q \leq |\mathcal{M}| \exp(-\frac{\q^2}{10\n^22^\n}),
\]
where we used Claim~\ref{cl:ShadowUnion} to justify the first inequality.

Taking logarithms on both sides yields the desired lower bound $\log 
|\mathcal{M}|= \Omega(\frac{\q^2}{\n^2 2^\n})$.
\qquad \end{proof}

As remarked earlier, the above communication lower bound implies that
the space required by a deterministic streaming algorithm to find a
valid coloring is exponential in $\n$ even for hypergraphs that have
very simple randomized coloring strategies; however it does not yield any such
lower bound for hypergraphs that have fewer than $2^{\n/2}$ hyperedges.  In
order to overcome this limitation, we generalize the analysis above
to the $k$-player hypergraph coloring problem. We will see later that for
$\q \leq 2^{\n/2}$ there do exist efficient two-player one-round protocols, so we
could not have proved our lower bounds while restricting our attention to the
two-player setting.

\subsubsection*{Proof of Theorem \ref{thm:result1}:}

We first prove the following theorem:

\begin{theorem}
Let $k \geq 1$ and $(k+1)\q \leq \m(\n)$ .
\[ \CC{1}(\Hgame(\V,\q,k+1)) = \Omega\left(\q \left[\frac{\q}{\V}\frac{{\V/2 \choose \n}}{{\V \choose 
\n}}\right]^\frac{1}{k}\right)
\]
\end{theorem}

\smallskip
\begin{proof}
Recall that the protocol has $k+1$ players, $P_1,P_2,\ldots,P_k,
P_{k+1}$.  player $P_i$ receives a hypergraph $H_i$ with $\q$
hyperedges over the vertex set $[v]$. The communication starts with
$P_1$, who sends a message $m_1$ of length $\ell_1$ to player $P_2$;
in the $i$-th step, $P_i$ sends a message of length $\ell_i$ to
$P_{i+1}$. In the end, $P_{k+1}$ produces a valid two-coloring for the
hypergraph $H = H_1 \cup H_2 \cup \cdots \cup H_{k+1}$. Note that this is always possible since $H$ has at most $\m(\n)$ hyperedges, and hence, is two-colorable. It will be
convenient to view this coloring as a message sent by player
$P_{k+1}$, and set $\ell_{k+1}=v$ (the number of bits needed to
describe a coloring).

In a $(k+1)$-player protocol, after the messages sent by the first $i$
players have been fixed, we have a list of colorings that may still be
output at the end; we use the following notation to refer to this list
(recall that $m_{k+1}$ is a coloring):

\begin{eqnarray*}
 L(m_1,m_2,\ldots,m_i) &=&  \{m_{k+1}: \mbox{ for some input
   $H_1,H_2,\ldots,H_{k+1}$ the transcript is of the form} \\
 &&      \hspace{1in} (m_1,m_2,\ldots,m_i, \ldots,m_k,m_{k+1})\}. 
\end{eqnarray*}

In particular, by considering the situation at the beginning of the
protocol (when no messages have yet been generated), we let $L_0=
\{\chi: m_{k+1}=\chi \text{ is output by the protocol on some
  input}\}$.  
Let $s_0 = |\Delta(L_0)|/ {\V \choose \n }$, and for $i=1,\ldots, k+1$, let
\[ s_i(m_1, \ldots, m_i) = 
 \left. {|\Delta(L(m_1,m_2,\ldots,m_i))|} \middle/ {\V \choose \n}
   \right. \] and  
   
\[ s_i = \min_{(m_1,m_2,\ldots,m_i)} s_i(m_1, \ldots, m_i)
. \] 
   
   Here the minimum is taken over all possible sequences of first $i$
   messages that arise in the protocol.  In particular, $s_0$
   corresponds to the union of shadows of all colorings ever output by the
   protocol, and $s_{k+1}$ corresponds to the shadow of the output
   corresponding to the transcript (that is, the shadow of the last
   message, which is a coloring).

\begin{claim}
\label{claim:GeneralShadow}
\begin{eqnarray}
 s_0 & \leq &  1  \label{cl:a}\\
\forall i\in\{0, \ldots,  k\}: \quad s_i &\geq& s_{i+1}
\frac{\q}{\ell_{i+1}} \label{cl:b}\\
s_{k+1}&\geq& \frac{{\lceil{\V/2}\rceil \choose \n}}{{\V \choose \n}}. \label{cl:c}
\end{eqnarray}
\end{claim}

\begin{proof}
Inequality (\ref{cl:a}) is immediate from the definition of $s_0$. 

For Claim (\ref{cl:b}), we use ideas similar to those used in the
proof of Claim~\ref{cl:ShadowUnion} for two-player protocols. Fix
$m_1,m_2, \ldots, m_i$. We show that
$s_i(m_1,m_2,\ldots,m_i)$ is at least the right hand side of
(\ref{cl:b}).  Pick a random hypergraph $H$ (which we consider as
a possible input to $P_{i+1}$) as follows: choose $\q$ hyperedges
independently and uniformly from the set $\Delta(L(m_1,\ldots ,
m_i))$. When $H$ is presented to $P_{i+1}$, it must respond with a
message $m_{i+1}$. None of the colorings that $P_{k+1}$ produces after
that can include any hyperedge of $H$ in its shadow, that is, $\Delta(L(m_1,
\ldots, m_{i+1})) \cap H = \emptyset$. For each valid choice $m$ for
$m_{i+1}$, we have $\Delta(L(m_1, \ldots, m_i, m))
\subseteq \Delta(L(m_1, \ldots, m_i))$ and
\[ \Pr[H \cap \Delta(L(m_1, \ldots, m_i, m))=\emptyset] \leq 
\left(1-\frac{s_{i+1}(m_1,\ldots, m_i, m)}{s_{i}(m_1, \ldots, m_i))}\right)^{\q}
\leq \left(1-\frac{s_{i+1}}{s_{i}(m_1, \ldots, m_i))}\right)^{\q}.
\]
Thus,
\[
 1 = \Pr_{H}[\exists m_{i+1}: \Delta(L(m_1, \ldots, m_{i+1})) \cap H =
   \emptyset] \leq 2^{\ell_{i+1}}\left(1-\frac{s_{i+1}}{s_{i}(m_1,
   \ldots, m_i))}\right)^{\q}
\]
This yields $\exp(\ell_{i+1}-\frac{s_{i+1}\q}{s_{i}(m_1, \ldots, m_i)})
\geq 1$, giving $s_{i}(m_1, \ldots, m_i) \geq s_{i+1}
\left(\frac{\q}{\ell_{i+1}}\right)$.  By minimizing over valid
sequences $(m_1,\ldots, m_i)$, we justify our claim.

Claim \ref{cl:c} follows from the fact the shadow of every
coloring has at least ${\lceil{\V/2}\rceil \choose \n}$ hyperedges.
\qquad \end{proof}

By combining parts (\ref{cl:a}) and (\ref{cl:b}), we obtain
\[ 1 \geq \left.{\q^{k+1}} s_{k+1} \middle/ {\prod_{i=1}^{k+1}\ell_i} \right. .\]
The theorem follows from this by using (\ref{cl:c}),  noting that
$\ell_{k+1} \leq \V$, and for $i=1,2, \ldots,k$: $\ell_i \leq \max_{i \in [k]} \ell_i$.
\qquad \end{proof}

\begin{cor}
[Restatement of Theorem~\ref{thm:result1}] Every one-pass
deterministic streaming algorithm to two-color an $\n$-uniform
hypergraph with at most $\q$ hyperedges requires $\Omega(\min \inbrace{\frac{\q}{\n},\frac{\m(\n)}{\n}} )$ bits
of space.
\end{cor}

\begin{proof} Setting $k=\n$ in
  the previous result immediately yields that for hypergraphs on
  $\V=\n^2$ vertices and at most $\q=(\n+1)\q'$ hyperedges, the
  communication required is $\Omega(\q')$. The lower bound for
  streaming algorithms then follows from
  Proposition~\ref{prop:streamingtocommunication}.
\qquad \end{proof}



\subsection{Deterministic communication protocols for the two-coloring problem}
\label{sec:detub}
In the previous section, we derived our lower bound for the
deterministic streaming algorithms by invoking {\em multi-player}
communication complexity, because the two-player lower bound did not
give us a non-trivial lower bound for hypergraphs with fewer than
$2^{\n/2}$ hyperedges. In this section, we first show an upper bound in
the two-player setting, which shows that it was essential to consider
the multi-player setting in order to get the stronger lower
bound. Next, we consider two-round protocols, for they are related to
two-pass streaming algorithms. We show below, perhaps surprisingly,
that the two-round two-player deterministic communication complexity
for the problem is {\em $\poly(\n)$}. However, we do not have a
streaming algorithm with a matching performance.

\smallskip
\begin{theorem}
$\CC{1}(\Hgame(\V, 2^{\n/2}, 2)) = \poly(\n)$
\end{theorem}

\begin{proof}
Alice and Bob will base their protocol on a special collection of
lists of colorings $\mathcal{L}=\{L_1, \ldots, L_r \}$, with
$r=2^{\n}$.  Suppose Alice's hypergraph is $H_A$ and Bob's hypergraph
is $H_B$. The protocol will have the following form.
\begin{description}
 \item[Alice] Alice sends an index $i$ of a list $L_i \in \mathcal{L}$
   such that every coloring in $L_i$ is valid for $H_A$.
 \item[Bob]  Bob outputs a coloring $\chi \in L_i$ that is valid for $H_B$.
\end{description}

We next identify some properties on the collection of lists that
easily imply that the protocol above produces a valid two-coloring.

\begin{definition}[Good Lists]

\begin{enumerate}
\item[(a)] A collection of lists $\mathcal{L}$ is good for Bob, if for every
hypergraph $H_B$, in \underline{every} list $L$ in $\mathcal{L}$ there is
a valid coloring for $H_B$ in $L$.

\item[(b)] A collection of lists $\mathcal{L}$ is good for Alice, if
  for every hypergraph $H_A$, there is \underline{some} list $L$ in
  $\mathcal{L}$, such that \underline{every} coloring in $L$ is valid
  for $H_A$.

\end{enumerate} 
\end{definition}
Note that good collections are defined differently for Alice and Bob.
With this definition, the existence of a collection of lists that is
good for both Alice and Bob would furnish a one-way communication
protocol with $\lceil \log r \rceil$ bits of communication, and establish
the theorem.  Such a {\em good} collection is shown to
exist in Lemma~\ref{GoodListsExist} below.  \qquad \end{proof}

\begin{lemma}
\label{GoodListsExist}
Suppose $H_A$ and $H_B$ are restricted to have at most $2^{\n/2}$
hyperedges.  Then, there exists a collection of $\mathcal{L}=\{L_1,
\ldots, L_r\}$ of $r=2^{\n}$ lists, each with $k=\lceil 6\cdot
2^{\n/2}\log \V\rceil$ colorings, which is good for Alice and Bob.
\end{lemma}

\begin{proof}
We pick $\mathcal{L}$ by picking $r$ lists randomly: list $L_i$ will
be of the form $\{\chi_{i1},\ldots,\chi_{ik}\}$, where each coloring
is chosen independently and randomly from the set of all colorings. We
will show that with positive probability $\mathcal{L}$ is a good
collection of lists. To do this, we separately bound the probability that
$\mathcal{L}$ fails to satisfy conditions (a) and (b) above. We may
restrict attention to hypergraphs that have exactly $\q =
\lfloor{2^{\n/2}}\rfloor$ hyperedges without loss of generality, by considering the given hypergraph as a sub-hypergraph of one with $\lfloor{2^{\n/2}}\rfloor$ hyperedges .
\begin{enumerate}
\item[(a)] We first bound the probability that $\mathcal{L}$ is not good for Bob. Consider one list $L_i \in \mathcal{L}$. 
\begin{eqnarray*}
\Pr[L_i \text{ is not good for Bob }] & =& \Pr[\exists H_B \text{ such that none of 
the }k \text{ colorings are valid for } H_B] \\
 & \leq &  {{\V \choose \n} \choose
 \q}\left(\frac{\q}{2^{\n -1}}\right)^{k}.
\end{eqnarray*}
Since $\mathcal{L}$ is a collection of $r$ lists, we have 
\begin{eqnarray}
\Pr[\mathcal{L} \text{ is not good for Bob}] &=& \Pr[\exists L_i \in \mathcal{L} \text{ is not good for Bob}]  \nonumber \\
&\leq & r {{\V \choose \n} \choose \q}\left(\frac{\q}{2^{\n-1}}\right)^{k} \\
&\leq & 2^{ (\n\q\log \V + \log r) -k (\n - \log \q -1) }. \label{ineq:notgoodforBob}
\end{eqnarray}
To make the right hand side of (\ref{ineq:notgoodforBob}) at most $1/4$ (say), we
 later set $k$ to satisfy
\begin{equation}
k \gg  \frac{\n\q \log \V + \log r}{\n - \log \q -1}. \label{eq:boundonk}
\end{equation}
\item[(b)] We next bound the probability that $\mathcal{L}$ is not
  good for Alice. Fix a hypergraph $H_A$ for Alice, and consider $L_i
  \in \mathcal{L}$.
\begin{eqnarray*}
\Pr_{L_i}[\text{every $ \chi \in L_i$ is valid for }H_A] & \geq &  
\left(1-\frac{\q}{2^{\n-1}}\right)^{k}\\
 &\geq&  \exp\left(-\frac{\q k}{2^{\n -1}-\q}\right).
\quad 
\mbox{(using $1-x \geq \exp\left(-\frac{x}{1-x}\right)$)}
\end{eqnarray*}
Since $\mathcal{L}$ is a collection of $r$ such lists chosen independently
and there are ${{\V \choose \n} \choose \q}$ choices for $H_A$, we have 
\begin{eqnarray}
\Pr[\mathcal{L} \text{ is not good for Alice}] & \leq & 
{{\V \choose \n} \choose \q}
\left(1- \exp\left(-\frac{\q k}{2^{\n -1}-\q}\right)\right)^r  \nonumber \\
&\leq& \exp\left( \q n \ln \V - r \exp\left(-\frac{\q k}{2^{\n-1}- \q}\right)\right).
\label{ineq:notgoodforAlice}
\end{eqnarray}
To make the right hand side of (\ref{ineq:notgoodforAlice}) small (less than $1/4$, say), we will choose $r$ such that 
\begin{equation}
 r \gg  (\q\n \ln \V) \exp\left(\frac{\q k}{2^{\n -1}-\q}\right). \label{eq:boundonr}
\end{equation}
\end{enumerate}
Now, suppose  $\V= \poly(\n)$.  Then, one can
verify that if we set $k=\lceil 6\q \log \V \rceil$ and $r=2^{\n}$, then
(\ref{eq:boundonk}) and (\ref{eq:boundonr}) both hold for all large
$\n$ (for $q = \lfloor 2^{n/2} \rfloor$).  It follows that the required collection of lists exists.
\qquad \end{proof}


\subsubsection{Circuit upper bounds and a $2$-round communication protocol}

Next, we present our efficient two-player, two-round communication protocol for
hypergraph two-coloring.

\begin{theorem} \label{thm:two-round-ub}
$\CC{2}(\Hgame(\V, 2^n/8, 2)) = \poly(n)$
\end{theorem}

We prove Theorem~\ref{thm:two-round-ub} by exploiting the connection between circuit complexity and \emph{Karchmer-Wigderson games} \cite{KarchmerW90}. 

\medskip
\begin{definition}[Karchmer-Wigderson game]
Given a monotone boolean function $f:\zo^N \rightarrow \zo$, the Karchmer-Wigderson communication game $G_f$ between two players Alice and Bob is the following: Alice gets an input $x\in f^{-1}(0)$, and Bob gets an input $y\in f^{-1}(1)$, and they both know $f$. The goal is to communicate and find an index $i\in [N]$ such that $x_i=0$ and $y_i=1$ (such an index exists since $f$ is monotone). The communication complexity of this game is denoted by $\CC{}(G_f)$ (or specifically $\CC{r}(G_f)$ for an optimal $r$-round protocol).
\end{definition}
\medskip

Karchmer and Wigderson make the following connection:

\begin{theorem}[\cite{KarchmerW90}]\label{thm:KW-game}
If a monotone boolean function $f$ has a depth-$d$, size $s$ circuit, then $\CC{d-1}(G_f) = O(d \log s)$.
\end{theorem}

\smallskip

The proof of Theorem~\ref{thm:two-round-ub} is based on the circuit complexity of \emph{Approximate Majority functions}. Define the partial function $\AppMaj_N$ on a subset of $\zo^N$ as follows: $\AppMaj_N(x)= 1$  when $|x|\geq 2N/3$, and $\AppMaj_N(x)= 0$  when $|x|\leq N/3$. A function $f:\zo^N \rightarrow \zo$ is an Approximate Majority function if it satisfies: $f(x)= 1$ when $|x| \geq 2N/3$, and $f(x)=0$ when $|x| \leq N/3$. We need the following Lemma.

\smallskip
\begin{lemma} \label{lem:reduction-to-maj}
For any Approximate Majority function $f:\zo^{2^\V} \rightarrow \zo$, we have:
\[
\CC{r}\left(\Hgame(\V, ~2^{\n}/8,~ 2)\right) \leq \CC{r}(G_f) 
\]
\end{lemma}

\begin{proof}
We prove this by reducing  $\Hgame(\V, ~2^{\n}/8,~ 2)$ to $G_f$, for any Approximate Majority function $f$. For convenience, let $N \coloneqq 2^\V$. Suppose Alice receives $H_A = (V, E_A)$, and Bob receives $H_B = (V, E_B)$. Let $L_A \in \zo^{N}$ be an indicator vector for Alice that marks which of the $N$ possible two-colorings of $V$ (arranged in some canonical order) are valid for $H_A$, and let $L_B$ be the corresponding indicator vector for Bob. Since we know that both $H_A, H_B$ have $\q \leq 2^\n/8$ hyperedges, a randomly chosen coloring will color either hypergraph with probability $\geq 3/4$, which implies $|L_A| \geq 3N/4$ and  $|L_B| \geq 3N/4$. To find a coloring valid for both $H_A$ and  $H_B$, Alice and Bob just have to find an index $i\in [2^\V]$ such that $L_A(i)= L_B(i)=1$.

Alice now complements her input to consider $\overline{L_A}$. Since $|\overline{L_A}| \leq N/4$, $\AppMaj_N(\overline{L_A}) = 0$, whereas $\AppMaj_N(L_B)=1$. Clearly, playing the Karchmer-Wigderson game for any Approximate Majority function $f$ on $N$ bits, $\overline{L_A}, L_B$ would find an index $i$  where $L_A = L_B = 1$, and consequently yield a valid coloring for $H_A \cup H_B$.
\end{proof}

Hence, a small depth-3 monotone circuit that computes $\AppMaj_{2^\V}$ on its domain would yield an efficient 2-round protocol for $\Hgame(\V, ~2^{\n}/8,~ 2)$, using Theorem~\ref{thm:KW-game}. Ajtai~\cite{Ajtai83} showed that such a circuit exists, and Viola~\cite{Viola09} further showed uniform constructions of such circuits.

\begin{theorem} \label{thm:appmaj-circuit}\cite{Viola09}
There exist monotone, uniform $\poly(N)$-sized depth-3 circuits for $\AppMaj_N$ on $N$ input bits.
\end{theorem} 

\begin{proof}[Proof (of Theorem~\ref{thm:two-round-ub})]
The proof of Theorem~\ref{thm:two-round-ub} is now immediate from Theorems~\ref{thm:KW-game}, \ref{thm:appmaj-circuit} and Lemma~\ref{lem:reduction-to-maj}, since $\V$ is polynomial in $\n$.  
\end{proof}



\section{Streaming algorithms for hypergraph two-coloring}
\label{sec:streaming-alg} 

In this section, we present streaming algorithms that come close to
the performance of the randomized off-line algorithm of Radhakrishnan
and Srinivasan~\cite{Radhakrishnan-Srinivasan}. We first point out why
this algorithm, as stated, cannot be implemented with limited memory.
Next, we show how an alternative version can be implemented using a
small amount of memory. This modified version, however, does return
colorings that are not valid (though it does this with small
probability); we show how by maintaining a small amount of additional
information, we can derive an algorithm that returns and valid coloring
with high probability or returns failure, but never returns an invalid
coloring. 

In order to describe our algorithms, it will be convenient to use $u, w$ to denote vertices in the hypergraph. Also, in this section, we  explicitly use $|V|$ to denote the size of the vertex set $V$, and $|E|$ to denote the number of hyperedges in the hypergraph for clarity. The vertex set of the hypergraph is identified with $[|V|]$, and $B$ is used to denote the number of bits required to represent a single vertex in $V$.

\subsection{The delayed recoloring algorithm}
\label{sec:off-line-delayed-recoloring}

Radhakrishnan and Srinivasan~\cite{Radhakrishnan-Srinivasan} showed
that by introducing delays in the recoloring step of an algorithm
originally proposed by Beck~\cite{Beck} one can two-color hypergraphs
with more hyperedges than it was possible before.
\begin{theorem}[\cite{Radhakrishnan-Srinivasan}] \label{thm:radhakrishnan-srinivasan}
 Let $H=(V,E)$ be an $\n$-uniform hypergraph with at most $\frac{1}{10}
 \sqrt{\frac{\n}{\ln \n}} 2^\n$ hyperedges. Then $H$ is two-colorable; also
 a proper two-coloring can be found with high probability in time
 $O(\poly(|V| + |E|))$.
\end{theorem}

\begin{algorithm}
\caption{(Off-line delayed recoloring algorithm)}
\label{alg:off-line}
\begin{algorithmic}[1]
\STATE Input $H=(V,E)$.  
\STATE For all $u \in V$, independently set
$\chi_0(u)$ to $\mathrm{Red}$ or $\mathrm{Blue}$ with
probability $\frac{1}{2}$.

\STATE Let ${\cal M}_0$ be the set of hyperedges of $H$, that are monochromatic under $\chi_0$.

\STATE For all $u \in V$, independently set $b(u)$ to be $1$ with
probability $p$, and $0$ with probability $1-p$. 

\STATE Let $\pi$ be one of the $|V|!$  permutations of $V$, chosen uniformly at random.

\FOR{$i=1,2, \ldots, |V|$}

\STATE $\chi_i$ is obtained from $\chi_{i-1}$ by retaining the colors
of all vertices except perhaps $\pi(i)$. If $b(\pi(i))=1$, and some
hyperedge containing $\pi(i)$ was monochromatic under $\chi_0$ and all
its vertices still have the same color in $\chi_{i-1}$, then 
$\chi_{i-1}(\pi(i))$ is flipped to obtain $\chi_{i}(\pi(i))$.
\ENDFOR 
\STATE Output $\chi_f = \chi_{|V|}$.
\end{algorithmic}
\end{algorithm}
The term {\em off-line} refers to the fact that we expect all
hyperedges to remain accessible throughout the algorithm.  The term
{\em delayed recoloring} refers to the fact that vertices are
considered one after another, and are recolored only if the initially
monochromatic hyperedge they belonged to has not been set right by a vertex
recolored earlier.

For a coloring $\chi$, let ${\cal B}(\chi)$ be the set of hyperedges
that are colored entirely blue under $\chi$, and let ${\cal R}(\chi)$
be the set of hyperedges colored entirely red under $\chi$. In
\cite{Radhakrishnan-Srinivasan}, the following claims are
established for this algorithm under the assumption that $|E| \leq
\frac{1}{10}\sqrt{\frac{\n}{\ln \n}} 2^n$. In the following, the
parameter  $p$ denotes
the probability that the bit $b(v)$ is set to $1$ in Algorithm \ref{alg:off-line}.
\begin{eqnarray}
\Pr[({\cal B}(\chi_f) \cap {\cal B}(\chi_{0}) \neq \emptyset) 
       \vee
    ({\cal R}(\chi_f) \cap {\cal R}(\chi_{0}) \neq \emptyset) ]
&\leq & \frac{2}{10} \sqrt \frac{\n}{\ln \n} (1-p)^\n \label{ineq:staysmono}
\\
\Pr[{\cal B}(\chi_f) \setminus {\cal B}(\chi_0)  \neq \emptyset] 
&\leq& \frac{2\n p}{100 \ln \n};
\label{ineq:becomesblue}
\\
\Pr[{\cal R}(\chi_f) 
\setminus {\cal R}(\chi_0) \neq \emptyset] 
&\leq& \frac{2\n p}{100 \ln \n}. \label{ineq:becomesred}
\end{eqnarray}

The first inequality (\ref{ineq:staysmono}) helps bound the
probability that an initially monochromatic hyperedge does not change
our attempts at recoloring; the second inequality
(\ref{ineq:becomesblue}) considers the event where an initially
non-blue hyperedge become blue because of recoloring; the third
inequality (\ref{ineq:becomesred}) similarly refers to the event where
an initially non-red hyperedge becomes red because of recoloring. These
events together cover all situations where $\chi_f$ turns out to be
invalid for $H$. Now, set $p=\frac{1}{2\n}(\ln \n - \ln\ln \n)$ so that we have
\begin{equation}
\frac{2}{10}  \sqrt \frac{\n}{\ln \n} (1-p)^\n + \frac{4\n p}{100 \ln \n}
\leq \frac{1}{5} + \frac{1}{50} \leq 
\frac{11}{50}. \label{eq:offlineerror}
\end{equation}
Thus, the above algorithm produces a valid two-coloring with
probability at least $\frac{1}{2}$. Furthermore, since the entire
hypergraph is available in this off-line version, we can efficiently verify
that the final coloring is valid. By repeating the algorithm
$\lceil{\log(1/\delta)}\rceil$ times, we can ensure that the algorithm
produces a valid coloring with probability at least $1-\delta$, and
declares failure otherwise (but never produces an invalid coloring).

\begin{remark}
As stated in \cite{Radhakrishnan-Srinivasan}, for large enough $\n$, the bound on the number of hyperedges in the above theorem can be improved to $0.7\sqrt{\frac{\n}{\ln \n}} 2^\n$ (with the same algorithm); this carries over to our results as well. For clarity, we state and use the uniform bound of $\frac{1}{10}\sqrt{\frac{\n}{\ln \n}} 2^\n$ that holds for all $n\geq 2$.
\end{remark}

\paragraph{Implementation in the streaming model:}
Now, suppose the hyperedges arrive one at a time, each hyperedge is a
sequence of $\n$ vertices, each represented using $B$ bits. In time
$O(B\n)$ per hyperedge, we may extract all the vertices. The corresponding
colors and bits for recoloring can be generated in constant time per
vertex. Generating the random permutation $\pi$ is a routine matter,
we maintain a random permutation of the vertices received so far by
inserting each new vertex at a uniformly chosen position in the
current permutation; using specialized data structures, the total time (in the RAM model, assuming a vertex description fits into a word)
taken for generating $\pi$ is $O(|V|\log|V|)$. 
Overall, the the algorithms can be implemented in time
$\tilde{O}(\n |V||E|)$ (the $\tilde{O}$ hides some
$\mathrm{polylog}(|V|)$ factors). As described, the space required is rather large
because we explicitly store the hyperedges.  Observe, however, that we
need store only the hyperedges that are monochromatic after the
initial random coloring. Thus, the algorithm can be implemented using
$\tilde{O}(|V| + \n  \sqrt{\frac{\n}{\ln \n}})$ bits of space on the
average.  The space and time requirements may then be considered
acceptable, but this implementation does not allow us to determine if
the coloring produced at the end is valid; moreover, it is not clear
how we may reduce the error probability by repetition. In the next
section, we give an implementation that does not suffer from this
deficiency.

\subsection{An efficient streaming algorithm}
\label{randomizedonepass}

In this section, we modify the randomized streaming algorithm of the
previous section so that it uses $O(|V|B)$ bits of space and takes
time comparable to the off-line algorithm. This version maintains a
running list of vertices as suggested above, but it assigns them
colors immediately, and with the arrival of each hyperedge decides if some
of its vertices must be recolored. The algorithm is derived from the
off-line algorithm. It maintains enough information to ensure that the
actions of this streaming algorithm can somehow be placed in
one-to-one correspondence with the actions of the off-line algorithm.

\begin{algorithm}
\caption{Delayed recoloring with limited space}
\label{alg:streaming}
\begin{algorithmic}[1]
\STATE Input: The hypergraph $H=(V,E)$ as a sequence of hyperedges $h_1, h_2, \ldots$.
  
\STATE {\em The algorithm will maintain for each vertex $u$ three
  pieces of information: 
\begin{enumerate}
\item[(i)] its initial color $\chi_0(u)$; 
\item[(ii)] its current color $\chi(u)$;
\item[(iii)] a bit $b(u)$ that is set to $1$ with probability $p$;
\end{enumerate}
Furthermore, it maintains a random permutation $\pi$ of the vertices received so far.}
\FOR{$i=1,2,\ldots, |E|$} 
\STATE Read $h_i$. 
 \FOR {$u \in h_i$}
  \IF{$u$ has not been encountered before}
   \STATE Set $\chi_0(u)=\chi(u)$ to $\mathrm{Red}$ or $\mathrm{Blue}$ with probability $\frac{1}{2}$.
         Set $b(u)$ to $1$ with probability $p$ and $0$ with
         probability $1-p$.
         Insert $u$ at a random position in $\pi$. \label{linegeneratechi}
  \ENDIF
 \ENDFOR
  \IF{$h_i$ was monochromatic under $\chi_0$ and all its vertices have the same color under $\chi$}
    \STATE   Let $u$ be the first vertex (according to the current $\pi$) such that $b(u)=1$.
        \IF{$\chi_0(u)=\chi(u)$} \STATE Flip $\chi(u)$.
        \ENDIF
  \ENDIF
\ENDFOR
\STATE Output $\chi_f = \chi_{|V|}$.
\end{algorithmic}
\end{algorithm}
We wish to show that Algorithm~\ref{alg:streaming} succeeds in constructing
a valid coloring with probability at least $\frac{1}{2}$. To justify this,
we compare the actions of this algorithm and the off-line delayed
recoloring algorithm stated earlier, and observe that their outputs
have the same distribution.

\newcommand{\hchi}{\hat{\chi}}
\newcommand{\hb}{\hat{b}}
\newcommand{\hpi}{\hat{\pi}}

Recall that $\chi_0$ was the initial coloring generated in the
Algorithm~\ref{alg:off-line}. Let $\hchi_0$ be the corresponding
coloring for Algorithm~\ref{alg:streaming}: that is, let $\hchi_0(u)$
be the color $u$ was first assigned when the first hyperedge
containing $u$ appeared in the input. Similarly, let $\hb$ be the
sequence of bits generated by the above algorithm, and $\hpi$ be the
final permutation of vertices that results.  Notice $(\chi_0, b, \pi)$
and $(\hchi_0,\hb,\hpi)$ have the same distribution.

In Algorithm~\ref{alg:off-line}, once $\chi_0$, $b$ and $\pi$ are fixed,
the remaining actions are deterministic. That is,
$\chi_f=\chi_f(\chi_0,b,\pi)$ is a function of $(\chi_0, b, \pi)$.
Similarly, once we fix (i.e., condition on) $(\hchi_0,\hb,\hpi)$, the
final coloring $\hchi_f=\hchi_f(\hchi_0,\hb,\hpi) $ is fixed.

\begin{lemma}
\label{samecoloring}
For all $\chi$, $b$ and $\pi$, we have $\chi_f(\chi_0,b,\pi) = \hchi_f(\chi_0,b,\pi)$.
\end{lemma}
\begin{proof}
Suppose $\chi_f(u) \neq \hchi_f(u)$ for some vertex $u$. One of them
must equal $\chi_0(u)$. We have two cases.  

\begin{description}
\item[Case 1]: $\chi_{0} (u)=\hchi_f(u) \neq {\chi_f} (u)$. Suppose that $h$ is a hyperedge that necessitated $u$'s
  recoloring in the off-line algorithm. This implies that $b(u)=1$,
  and $b(w)=0$ for the vertices $w$ of $h$ that appeared before $u$ in
  $\pi$. Now in the streaming algorithm above, when $h$ is considered,
  we would find that $h$ is monochromatic in $\chi_0$, and flip
  $\chi(u)$ unless it is already flipped. Once flipped, the color of
  $u$ will not change again. Thus, $\hchi_f(u) \neq \chi_0(u)$---a
  contradiction.

\item[Case 2]: $\chi_{0} (u) = {\chi_f} (u) \neq {\hchi_f} (u)$.  Let
  $h$ be a hyperedge that necessitated the recoloring of $u$ in the
  above streaming algorithm. This implies that $b(u)=1$, and all
  vertices $w$ of $h$ that appear before $u$ in permutation $\pi$ have
  $b(w)=0$. But in such a situation, the original off-line algorithm
  will find $h$ monochromatic when $u$ is considered, and will flip
  its color.  Thus, $\chi_0(u) \neq \chi_f(u)$---a contradiction. 
\end{description}
\end{proof}

From the above lemma, we conclude that (\ref{eq:offlineerror}) applies to Algorithm~\ref{alg:streaming} as well.

\begin{cor}[Compare Theorem~\ref{thm:radhakrishnan-srinivasan}]
Let $H=(V,E)$ be an $\n$-uniform hypergraph with at most $\frac{1}{10}
\sqrt{\frac{\n}{\ln \n}} 2^\n$ hyperedges. Then, with probability at least
$\frac{1}{2}$ Algorithm~\ref{alg:streaming} produces a valid
two-coloring for $H$.
\end{cor}

 Note, however, that Algorithm~\ref{alg:streaming} lacks the desirable
property that it outputs only valid colorings. A straight-forward
check does not seem possible using $O(|V|B)$ workspace.  We show below
that by carefully storing some vertices of a few hyperedges, we can
in fact achieve this using only $O(\n B)$ workspace.

\subsection{An efficient algorithm that produces only valid colorings}

As remarked before the events considered in inequalities
(\ref{ineq:staysmono}), (\ref{ineq:becomesblue}) and
(\ref{ineq:becomesred}) cover all situations when the algorithm might
return an invalid coloring. So, in order never to return an invalid
coloring, it suffices to guard against these events.
\begin{enumerate}
\item[(a)] To ensure that some vertex in every initially monochromatic
  hyperedge $h$ does change its color, we just need to verify that $b(u)=1$ 
  for some $u \in h$. This we can ensure by examining $\chi_0$ and $b$ just after they 
  are generated in line~\ref{linegeneratechi} of Algorithm~\ref{alg:streaming}. 
\item[(b)] To ensure that no initially non-blue  hyperedge $h$ has turned blue, we will
will save the red vertices of $h$ whenever there is potential for them all turning blue, and
guard against their turning blue. Consider the hypergraph $H_{\mathrm{Blue}}$ with hyperedges
\[ E(H_{\mathrm{Blue}}) = 
\{ h \cap \chi_0^{-1}(\mathrm{Red}): h \cap \chi_0^{-1}(\mathrm{Red}) \neq \emptyset 
 \mbox{ and $b(u)=1$ for all $u \in h \cap \chi_0^{-1}(\mathrm{Red})$}\}.\]
We will then verify that no hyperedge in $H_{\mathrm{Blue}}$ is entirely blue in the end.

\item[(c)] To ensure that no initially non-red hyperedge $h$ has
  turned red, we consider the corresponding hypergraph $H_{\mathrm{Red}}$ with hyperedges
\[ E(H_{\mathrm{Red}}) =
\{ h \cap \chi_0^{-1}(\mathrm{Blue}): h \cap
\chi_0^{-1}(\mathrm{Blue}) \neq \emptyset \mbox{ and $b(u)=1$ for all
  $u \in h \cap \chi_0^{-1}(\mathrm{Blue})$} \},\] and verify in the end
that no hyperedge in $H_{\mathrm{Red}}$ is entirely red in the end.
\end{enumerate}
It remains to show how $H_{\mathrm{Blue}}$ and $H_{\mathrm{Red}}$ can
be stored efficiently. Each hyperedge will be stored separately by listing all its vertices.
Thus, the expected sum of the sizes of the hyperedges in $H_{\mathrm{Blue}}$ is 
\newcommand{\E}{\mathbb{E}}
\newcommand{\size}{\mathrm{size}}
\begin{eqnarray*}
\E[\size(H_{\mathrm{Blue}})] &=&  \sum_{h \in E(H)} 2^{-\n} \sum_{i=1}^\n i {\n \choose i} {p}^i  \\
& \leq &  |E|~ \n p 2^{-\n} \sum_{i=0}^{\n-1} {{\n-1} \choose {i-1}} {p}^{i-1} \\
& \leq &  |E|~\n p 2^{-\n} (1+p)^{\n-1}.
\end{eqnarray*}
Clearly, the same bound applies to $H_{\mathrm{Red}}$. We will choose
the value of $p$, so that with high probability
$\size(H_{\mathrm{Blue}})$ and $\size(H_{\mathrm{Blue}})$ are both at
most $\n$ (if either of them exceeds $\n$, the algorithm terminates with
failure), while also ensuring that the right hand sides of
(\ref{ineq:staysmono}), (\ref{ineq:becomesblue}) and
(\ref{ineq:becomesred}) stay small. As before, set $p=\frac{1}{2\n}(\ln
\n - \ln\ln \n)$.  To the failures accounted for in
(\ref{eq:offlineerror}), we must account for
$\size(H_{\mathrm{Blue}})$ or $\size(H_{\mathrm{Blue}})$ exceeding
$\n$.  Now, $\E[\size(H_{\mathrm{Blue}})] \leq \frac{1}{10} \sqrt
  \frac{\n}{\ln \n} \frac{\ln \n}{2} \sqrt{\frac{\n}{\ln \n}} =
  \frac{\n}{20}$, and by Markov's inequality
  $\Pr[\size(H_{\mathrm{Blue}}) \geq \n] \leq \frac{1}{20}$.
  Similarly, $\Pr[ \size(H_{\mathrm{Red}}) \geq \n] \leq
  \frac{1}{20}$. 
Thus, the revised algorithm fails to deliver a valid coloring with
probability at most $\frac{11}{50} + \frac{1}{20} + \frac{1}{20} <
\frac{1}{2}$, while it uses at most $O(|V|B + \n B)$ bits of space.

\subsection{Special cases when $|E|$ is large}
The above streaming algorithm deals with the cases when an upper bound on the number of hyperedges of $H$ guarantees that it has a valid two-coloring. However, under certain conditions, hypergraphs with  many more hyperedges are also known to be two-colorable. We show that efficient streaming algorithms exist for two-coloring hypergraphs in some of these cases, as described in Theorem~\ref{thm:vertices-bounded} and Theorem~\ref{thm:LLL-streaming}.

\subsubsection{Coloring $n$-uniform hypergraphs with few vertices}
\label{boundeduniverse}

We now give the proof of Theorem~\ref{thm:vertices-bounded}. Recall that in this setting, the number of vertices of $H$ is bounded by $ n^2 / t$, and the number of hyperedges by $2^{n-1} \exp(t/8)$. Furthermore, the number of vertices is known to the algorithm in advance.
%

\begin{proof}[Proof of Theorem~\ref{thm:vertices-bounded}]
We prove the first part of the theorem by showing that a random, balanced two-coloring of the vertex set of $H$ is valid with non-zero probability. 

Suppose $V=[\V]$. Consider the coloring $\chi$ obtained by
partitioning $[\V]$ into two roughly equal parts and coloring one part
$\mathrm{Red}$ and the other part $\mathrm{Blue}$. Every hyperedge has
the same probability $p$ of being monochromatic under this coloring.
We first show the following:

\begin{claim}
$p \leq 2^{-(\n-1)} \exp\left(- \frac{(\n-1)^2}{2\V} \right)$.
\end{claim}

{Proof of claim:}
If $\V \leq 2(\n-1)$, then this probability is $0$. In general,
the probability that a hyperedge is monochromatic is 
\[
p = \frac{{{\lceil{\V/2}\rceil} \choose \n} + {{{\lfloor{\V/2}\rfloor} \choose \n}}} {{\V \choose \n}}
\leq \frac{2{{\V/2} \choose \n}}{{\V \choose \n}} 
\quad \mbox{(since ${a\choose \n}$ is convex for $a \geq \n-1$).}
\]
We then have
\begin{eqnarray*}
p \leq&  \frac{2 \cdot \frac{\V}{2}(\frac{\V}{2}-1)(\frac{\V}{2}-2) \cdots (\frac{\V}{2}-\n+1)}
                {\V(\V-1)(\V-1) \cdots (\V-\n+1)} 
=&   2^{-(\n-1)} \prod_{i=1}^{\n-1} \frac{\V-2i}{\V-i}\\
&&=~  2^{-(\n-1)} \prod_{i=1}^{\n-1} \left(1-\frac{i}{\V-i}\right)\\
&&\leq ~ 2^{-(\n-1)} \exp\left(- \sum_{i=1}^{\n-1} \frac{i}{\V-i}\right).
\end{eqnarray*}
Now,
\begin{eqnarray*}
- \sum_{i=1}^{\n-1} \frac{i}{\V-i} 
                             &=~ (\n-1) - \V \sum_{j=\V-\n+1}^{\V-1} \frac{1}{j}
                             &\leq~ (\n-1) - \V\int_{\V-\n+1}^{\V} \frac{1}{x} \mathit{d}x\\
                             & & \leq ~(\n-1) + \V \ln \left( 1 -\frac{\n-1}{\V} \right)\\
                        && \leq ~ (\n-1) - \V\left( \frac{\n-1}{\V}  + \frac{1}{2}(\frac{\n-1}{\V})^2 \right)\\
                        && =~  - \frac{(\n-1)^2}{2\V}.
\end{eqnarray*}
This shows that $p \leq 2^{-(\n-1)} \exp\left(- \frac{(\n-1)^2}{2\V} \right)$ and concludes the proof of the Claim.

\smallskip

To recover the first part of Theorem~\ref{thm:vertices-bounded}, we union-bound over the set of hyperedges, and observe that $\frac{(\n-1)^2}{v} \geq (\n-1)^2 \cdot \frac{t}{\n^2} \geq t/8$. 

\paragraph{A streaming algorithm:}
The corresponding streaming algorithm is essentially trivial given the proof above. Suppose we know that our hypergraph $H$ satisfies the conditions in theorem; i.e. $H$ has 
$\V=\frac{\n^2}{t}$ vertices and at most $2^{\n-2} \exp(\frac{t}{2})$
hyperedges. Since we know the number of vertices in advance, we just pick a random coloring $\chi$ that assigns colors
Red and Blue to roughly the same number of vertices, and verify that
no incoming hyperedge is monochromatic under $\chi$. The algorithm uses $O(\V)$ bits
of space, and returns a valid coloring with probability $\frac{1}{2}$
and fails otherwise. The probability of failure can be reduced to
$\delta$ by running the algorithm $O(\log \frac{1}{\delta})$ times in
parallel, and choosing as the final output any one of the colorings that are valid for $H$. Note that this increases the space requirements by a multiplicative factor of $O(\log(\frac{1}{\delta}))$.
\end{proof}

\begin{remark}
The following corollary can be deduced by setting $t = 8 \ln 2\n$.
\begin{cor}
Every $\n$-uniform hypergraph with at most $\frac{\n^2}{8 \ln 2\n}$
vertices and at most $\n^2 2^\n$ hyperedges is two-colorable.
\end{cor}

The corollary shows that Erd\H{o}s' construction~\cite{Erdos64} of non-two-colorable hypergraphs having $q = n^2 2^n$ hyperedges on a vertex set of size $\V = \Theta(n^2)$ is  nearly optimal in terms of the number of vertices.
\end{remark}

\subsubsection{When hyperedge intersections are bounded}

Using the Lov\'{a}sz Local Lemma (we state an algorithmic version due to Moser and Tardos that we use below; a weaker version first appeared in~\cite{Erdos-Lovasz}), \cite{Radhakrishnan-Srinivasan} show that $\n$-uniform hypergraphs where no hyperedge intersects more than $0.7 \sqrt{\frac{\n}{\ln \n}}~
2^\n$ others has a two-coloring (for $\n$ large enough). Note that  this does not require a bound on the number of hyperedges of $H$, but only one on the number of intersections of any one hyperedge with others.

The algorithms of Moser and Tardos~\cite{Moser-Tardos} can then be used to recover a valid two-coloring of $H$. As in the previous case, though the algorithm a-priori requires access to the entire hypergraph in an offline fashion, we can modify it to adapt it to the streaming setting. Our result, however, requires a slightly stronger bound on the number of intersections as compared to the offline algorithm in \cite{Radhakrishnan-Srinivasan}. We first briefly review the Moser-Tardos algorithm.

\paragraph{Moser-Tardos Algorithm for Lov\'{a}sz Local Lemma:}
Let $X$ be a finite set of events determined by a finite set $P$ of mutually independent random variables,
such that each event of $X$ is determined by a subset of the variables in $P$.
Let $G_X$ denote the dependency graph of the events, i.e.,
two events $A$ and $B$ are connected by an edge if and only if they share common variables. 
For any event $A \in X$, we denote by $N(A)$ the events which are neighbors of $A$ in $G_X$. 
An assignment or evaluation of the variables {\it violates} an event $A$ if it makes $A$ happen. The algorithm can be stated as follows:

\begin{enumerate}
\item[Step 1.] Evaluate each variable in $P$ independently at random.
 
\item[Step 2.] If there exists at least one violated event in $X$, construct a maximal independent set $M$ of the sub-graph (of $G_X$)
induced by the violated events in $X$.
Independently perform random re-evaluation of each variable that belongs to one of the events of $M$.

\item[Step 3.] If there are no violated events, output the current evaluation of the variables. Otherwise, go to Step 2.
\end{enumerate}
Moser and Tardos showed that if a certain `local' condition is assumed to hold for each of the events, one can bound the expected number of times Step $2$ of the algorithm is executed.

\begin{theorem} \label{thm:mosertardos}
\cite{Moser-Tardos} 
If $\epsilon >0$ and there exists real number assignments $x: X \rightarrow (0, 1)$ such that
\begin{equation}
\label{eq:assignment}
\forall A \in X: \Pr[A] \leq (1-\epsilon) x(A) \prod_{B \in N(A)} (1-x(B)),
\end{equation}
\noindent then the algorithm executes step $2$ an expected $O(\frac{1}{\epsilon} \log \sum_{A \in X} \frac{x(A)}{1-x(A)})$ number of times before it finds an evaluation 
of $P$ violating no event in $X$.  
\end{theorem}

\medskip

We now give the proof of Theorem~\ref{thm:LLL-streaming}. Since the proof closely follows that of Theorem~\ref{thm:mosertardos} and the algorithm is also similar; we only give an outline here. 

\begin{proof}[Proof Sketch of Theorem~\ref{thm:LLL-streaming}]
 For each hyperedge $h \in E$, let $X_h$ denote the event that $h$ is monochromatic in a $2$-coloring of $H$ and let $X=\{X_h: h \in H \}$.
Therefore, $\Pr[X_h] = 2^{1-n}$. 
If each hyperedge of $H$ intersects at most $\frac{(1-\epsilon) 2^{n-1}}{e} -1$
other hyperedges, we can assign $x(H_h)= \frac{e}{(1-\epsilon) 2^{n-1}}$ to satisfy Equation~\ref{eq:assignment} for all $X_h \in X$.
(We use $(1 - \frac{1}{r+1})^r \geq e^{-1}$ for any $r \geq 1$.)

Our streaming algorithm closely follows the Moser-Tardos algorithm stated above. In our setting, we start with a uniformly random coloring of the vertices.
Thereafter, Step $2$ of the algorithm can be executed once in each pass as follows:
whenever a hyperedge arrives, we mark and store its vertices in our workspace if and only if it is monochromatic in the current coloring and 
does not intersect with any of the previously marked hyperedges. Since at most $\frac{|V|}{n}$ disjoint hyperedges can be marked in one pass,
only $O(|V| B)$ bits of workspace is required to store all their vertices.
At the end of each pass, we randomly re-evaluate the colors of each of the marked vertices and start another pass over the data.
The algorithm stops if and only if there are no monochromatic hyperedges in the current coloring. By Theorem \ref{thm:mosertardos}, this algorithm terminates after
$O(\log \frac{|E|}{2^{n-1}}) = O(\log |V|)$ expected number of passes.  By Markov's inequality, for a large enough constant $C$, the algorithm produces a valid coloring in $C\cdot \log |V|$ passes with probability at least $3/4$.
\end{proof}

\section{Conclusion}

The lower bound we obtain on the space requirements of one-pass
streaming algorithms is optimal (up to $\poly(\n)$) factors.
We present an efficient two-player two-round deterministic
communication protocol for two-coloring $\n$-uniform hypergraphs with
up to $2^{\n}/8$ hyperedges, but we do not know if there is a corresponding
streaming algorithm. It is not difficult to
come up with a deterministic streaming algorithm that works in $|V|$
passes for the above setting using the method of conditional expectations: for some fixed ordering of $V$, each pass would count the expected number of hyperedges properly two-colored given a fixing of coloring of the previous vertices, and thereby determine the color of the last vertex considered. It would be interesting to know if one can do better. Even in the two-player communication setting it would be interesting to
determine if the protocol can accommodate up $\omega(2^n)$ hyperedges,
perhaps even $\Omega(\sqrt{\frac{\n}{\ln \n}} 2^\n)$ hyperedges.

Our two-coloring algorithm for hypergraphs with $O(\frac{\n^2}{t})$
vertices does not improve on the bound provided via the delayed
recoloring algorithm when $t$ is small, say, $o(\log \n)$. We believe
it should be possible to combine our argument and the delayed
recoloring algorithm to show that if the number of vertices is
$o(\n^2)$ then we can two-color hypegraphs with strictly more than
$\frac{1}{10}\sqrt{\frac{\n}{\ln \n}} 2^\n$ hyperedges.


{\small
\bibliographystyle{alpha}
\bibliography{hypergraph-bib}
}

\end{document}